\documentclass[11pt]{article}
\usepackage[utf8x]{inputenc}
\usepackage{algorithm,algorithmic,amsthm,amsmath}
\usepackage{multirow,booktabs,mathtools,amssymb}
\usepackage[square,comma,sort,numbers]{natbib}
\usepackage{fullpage}

\newtheorem{theorem}{Theorem}[section]

\newtheorem{claim}{Claim}[section]

\newtheorem{corollary}{Corollary}[section]
\newtheorem{definition}{Definition}

\newcommand{\inv}[1]{\frac{1}{#1}}
\newcommand{\cache}[1]{\mathbb{C}_{#1}}

\newcommand{\randmark}{{\sc RandomizedMarking}}
\newcommand{\zappingoff}{\mathcal{F}}
\newcommand{\rentalpagingcilp}{{\sc RentalPagingCILP}}
\newcommand{\rentalcachingcilp}{{\sc RentalCachingCILP}}
\newcommand{\rentalcachingmeta}{{\sc RentalCachingMeta}}
\newcommand{\zappingpagingcilp}{{\sc ZappingPagingCILP}}
\newcommand{\zappingcachingcilp}{{\sc ZappingCachingCILP}}
\newcommand{\rentalzappingpagingcilp}{{\sc RentalZappingPagingCILP}}
\newcommand{\rentalzappingcachingcilp}{{\sc RentalZappingCachingCILP}}
\newcommand{\rentalzappingcachingmeta}{{\sc RentalZappingCachingMeta}}

\newcommand{\ALG}{\mathrm{ALG}}
\newcommand{\ALGSR}{\mathrm{ALG_{SR}}}
\newcommand{\ALGINF}{\mathrm{ALG_{\infty}}}
\newcommand{\ALGC}{\mathrm{ALG_{C}}}
\newcommand{\ALGZ}{\mathrm{ALG_{Z}}}
\newcommand{\OPT}{\mathrm{OPT}}
\newcommand{\OPTSR}{\mathrm{OPT_{SR}}}
\newcommand{\OPTINF}{\mathrm{OPT_{\infty}}}
\newcommand{\OPTC}{\mathrm{OPT_{C}}}
\newcommand{\OPTZ}{\mathrm{OPT_{Z}}}

\title{File Caching with Rental Cost and Zapping}
\author{Monik Khare \hspace{1cm} Neal E. Young}
\date{}

\begin{document}

\maketitle

\begin{abstract}
The \emph{file caching} problem is defined as follows. Given a cache of size $k$ (a positive integer), the goal is to minimize the total retrieval cost for the given sequence of requests to files. A file $f$ has size $size(f)$ (a positive integer) and retrieval cost $cost(f)$ (a non-negative number) for bringing the file into the cache. A \emph{miss} or \emph{fault} occurs when the requested file is not in the cache and the file has to be retrieved into the cache by paying the retrieval cost, and some other file may have to be removed (\emph{evicted}) from the cache so that the total size of the files in the cache does not exceed $k$.

We study the following variants of the online file caching problem. \textbf{\emph{Caching with Rental Cost} (or \emph{Rental Caching})}: There is a rental cost $\lambda$ (a positive number) for each file in the cache at each time unit. The goal is to minimize the sum of the retrieval costs and the rental costs. \textbf{\emph{Caching with Zapping}}: A file can be \emph{zapped} by paying a zapping cost $N \ge 1$. Once a file is zapped, all future requests of the file don't incur any cost. The goal is to minimize the sum of the retrieval costs and the zapping costs.

We study these two variants and also the variant which combines these two (rental caching with zapping). We present deterministic lower and upper bounds in the competitive-analysis framework. We study and extend the online covering algorithm from \citep{young02online} to give deterministic online algorithms. We also present randomized lower and upper bounds for some of these problems.
\end{abstract}

\section{Introduction}

\subsection{Background}

The \emph{file caching} (or \emph{generalized caching}) problem is defined as follows. Given a cache of size $k$ (a positive integer), the goal is to minimize the total retrieval cost for the given sequence of requests to files. A file $f$ has size $size(f)$ (a positive integer) and retrieval cost $cost(f)$ (a non-negative number) for bringing the file into the cache. A \emph{miss} or \emph{fault} occurs when the requested file is not in the cache and the file has to be brought into the cache by paying the retrieval cost. When a file is retrieved into the cache, some other file may have to be removed (\emph{evicted}) from the cache so that the total size of the files in the cache does not exceed $k$. Weighted caching (or weighted paging) is the special case when each file has size 1. Paging is the special case when each file has size 1 and the retrieval cost for each file is 1.

An algorithm is \emph{online} if its response for each request is independent of all future requests. Let $\ALG (\sigma)$ be the cost of an algorithm ALG on request sequence $\sigma$, and let $\OPT(\sigma)$ be the corresponding optimal offline cost. $\ALG$ is $\alpha$-competitive if, for every request sequence $\sigma$, $\ALG(\sigma) \le \alpha \cdot \OPT(\sigma) + c$, where $c$ is a constant independent of the request sequence.

In this paper, we study the following variants of the file caching problem in the online setting using the competitive-analysis framework \citep{khare2012caching}.

\begin{definition}
	\label{dfn:rentalcaching}
	\textbf{\emph{Caching with Rental Cost (or Rental Caching)}}: There is a rental cost $\lambda$ (a positive number) for each file in the cache at each time step. The goal is to minimize the sum of the retrieval costs and the rental costs. In our mode, we allow time steps with no requests.
\end{definition}

\citet{chrobak2010sigact} proposes the rental caching problem and also presents some preliminary results. \emph{Weighted rental caching} (or, \emph{weighted rental paging}) is a special case of the rental caching problem where each file has size 1. \emph{Rental paging} is a special case where each file has size 1 and the retrieval cost for each file is 1.

The rental caching problem is motivated by the idea of energy efficient caching. Caching systems can save power by turning off the memory block that are not being used to store any files. Rental Caching models this by charging a rental cost for keeping each file in the cache. See \citep{salinger2012rental} for specific applications.

In section \ref{sec:rent-inf-cache} we show that the variant of rental caching where the cache has infinite size, is closely related to the ski-rental problem. The \emph{ski-rental} problem is the following. A pair of skis can be rented by paying $\$\lambda$ per day, or can be bought for the remainder of the ski season by paying $\$B$. It is not known when the season is going to end and the goal is to minimize the total money spent for the entire season \citep{karlin1988competitive}. 

\begin{definition}
	\label{dfn:cachingzapping}
	\textbf{\emph{Caching with Zapping}}: There is an additional cache of infinite size and any file can be added to this cache by paying a cost $N$ (a positive number greater than or equal to 1) at any time step. When a file is placed into this additional cache, we say the file has been zapped. A miss or fault occurs only when the requested file is not present in either cache. Thus, any future requests to a file in the additional cache do not incur any cost. The goal is to minimize the sum of the retrieval costs and the zapping costs.
\end{definition}

\emph{Weighted caching with zapping} (or, \emph{weighted paging with zapping}) is a special case of the caching with rental cost problem where each file has size 1. \emph{Paging with zapping} is a special case where each file has size 1 and the retrieval cost for each file is 1.

These variants generalize the file caching problem. File caching is a special case of rental caching where the rental cost is $0$. Similarly, caching is a special case of caching with zapping where the cost of zapping is arbitrarily large. We also study the variant which combines these two variants: \textbf{\emph{rental caching with zapping}}. In our model, there is no rental cost for files in the additional cache. Only the files in the size $k$ primary cache have to pay the rental cost.

\subsection{Previous work}

In 1985 \citet{sleator1985amortized} introduced the competitive-analysis framework. In \citep{sleator1985amortized} they show that the well-known paging rules like {\sc LeastRecentlyUsed} (LRU), {\sc FirstInFirstOut} (FIFO), and {\sc FlushWhenFull} (FWF) are $k$-competitive and that $k$ is the best ratio any deterministic online algorithm can achieve for the paging problem.

\citet{fiat1991marking} initiate the competitive analysis of paging algorithms in the randomized setting. They show a lower bound of $H_k$, where $H_k$ is the $k_{th}$ harmonic number, for any randomized algorithm. They give a $2H_k$-competitive \randmark\ algorithm. \citet{achlioptas2000competitive} show that the tight competitive ratio of \randmark\ is $2H_k - 1$. \citet{mcgeoch1991strongly} and \citet{achlioptas2000competitive} give optimal $H_k$-competitive randomized algorithms for paging.

For weighted caching, \citet{chrobak1990newresults} give a tight $k$-competitive deterministic algorithm. For the randomized case, \citet{bansal2007primal} give a tight $O(\log{k})$-competitive primal-dual algorithm.

For file caching, \citet{irani2002page} show that the offline problem is NP-hard. For the online case, \citet{irani2002page} give results for the \emph{bit model} ($cost(f) = size(f)$ for each file $f$) and \emph{fault model} ($cost(f) = 1$ for each file $f$). She shows that LRU is $(k+1)$-competitive for both models. \citet{cao1997cost} extend the result to file caching. \citet{young02online} independently gives {\sc Landlord} algorithm and shows that it is $k$-competitive for the file caching problem. \citet{irani2002page} gives an $O(\log^2{k})$-competitive randomized algorithm for bit and fault models. \citet{bansal2007primal} give an $O(\log{k})$-competitive randomized algorithm for both the models, and an $O(\log^2{k})$-competitive randomized algorithm for the general case.

\citet{young1994kserver} uses online primal-dual analysis to give a $k$-competitive deterministic online algorithm for weighted-caching. \citet{bansal2007primal, bansal2008randomized, buchbinder2009online} use online primal-dual approach to give randomized algorithms for the paging, weighted caching, and file caching problems. In a recent work, \citet{adamaszek2012log} builds on their
online primal-dual approach to give an $O(\log{k})$-competitive for the general case. In another recent work \citet{epstein2011variants} show that this online primal-dual approach can be extended to \emph{Caching with Rejection}. Caching with rejection is a variant of file caching where a request to a file, that is not in the cache, can be declined by paying a rejection penalty. In this variant, each request is specified as a pair $(f, r)$, where $f$ is the file requested and $r$ is the rejection penalty. Note that, caching with rejection is different from caching with zapping. In caching with zapping, a file can be zapped at any time step, while in caching with rejection, a file can be rejected only at the time step when it is requested. Moreover, a rejected file can incur retrieval cost or rejection penalty again in the future, while the zapped file does not incur any cost after it is zapped.

\citet{young12greedy} present a deterministic greedy $\Delta$-approximation algorithm for any covering problem with a submodular and non-decreasing objective function, and with arbitrary constraints that are closed upwards, such that each constraint has at most $\Delta$ variables. They show that their algorithm is $\Delta$-competitive for the online version of the problem where the constraints are revealed one at a time. Many online caching and paging problems reduce to online covering, and consequently, their algorithm generalizes many classical deterministic algorithms for these problems. These include LRU and FWF for paging, {\sc Balance} and {\sc Greedy Dual} for weighted caching, {\sc Landlord} (a.k.a. {\sc Greedy Dual Size}) for file caching, and algorithms for {\sc Connection Caching} \citep{young12greedy}. We study this approach and extend it to give deterministic online algorithms for the variants of online file caching studied in this paper.

\subsection{Our contributions}

We study rental caching, caching with zapping, and rental caching with zapping. We present deterministic and randomized lower and upper bounds for these new variants of paging, weighted caching, and caching in the online setting. We use the approach in \citep{young12greedy} to give deterministic algorithms for these online problems. While this approach is general, it doesn't necessarily give optimal online algorithms. The direct application of this approach yields sub-optimal algorithms in some of the cases we study in this paper. We describe these scenarios and also the appropriate modifications to the algorithm to achieve better competitive ratios.

\begin{table}[ht]\footnotesize
	\centering
	\label{tab:result-summary}
	\caption{Competitive ratios in this paper}
	\begin{tabular}[width=\textwidth]{| l *{4}{| c } | }
		\toprule
		\textbf{Problem} & & & \textbf{Lower Bound} & \textbf{Upper Bound} \\
		\toprule
		
		\multirow{7}{*}{Rental paging} & \multirow{4}{*}{Deterministic} & $\lambda \ge \inv{k}$ & $2 - \lambda$ & 2 \\
		\cmidrule{3-5}
		& & $\inv{k^2} \le \lambda < \inv{k}$ & \multirow{2}{*}{$\frac{k + k \lambda}{1 + k^2 \lambda}$} & $1+\inv{k\lambda}$ \\
		\cmidrule{3-3} \cmidrule{5-5}
		& & $\lambda < \inv{k^2}$&  & $k$ \\
		\cmidrule{2-5}
		& \multirow{2}{*}{Randomized} & $\lambda \ge \inv{k}$ & $\frac{e}{e-1}$ & $\frac{e}{e-1}$ \\
		\cmidrule{3-5}
		& & $\lambda < \inv{k}$ & $\frac{H_k + k^2 H_{k} \lambda}{1 + k^2 H_{k} \lambda}$ & $H_k + \frac{e}{e-1}$ \\
		\bottomrule

		\multirow{6}{*}{Weighted rental paging} & \multirow{3}{*}{Deterministic} & $\lambda \ge \inv{k}$ & $2 - \lambda$ & \multirow{3}{*}{$k$} \\
		\cmidrule{3-4}
		& & $\lambda < \inv{k}$ & $\frac{k + k \lambda}{1 + k^2 \lambda}$ & \\
		\cmidrule{2-5}
		& \multirow{2}{*}{Randomized} & $\lambda \ge \inv{k}$ & $\frac{e}{e-1}$ & $\frac{e}{e-1}$ \\
		\cmidrule{3-5}
		& & $\lambda < \inv{k}$ & $\frac{H_k + k^2 H_{k} \lambda}{1 + k^2 H_{k} \lambda}$ & $H_k + \frac{e}{e-1}$ \\
		\bottomrule

		Rental caching & \multicolumn{4}{|c|}{Same as weighted rental paging} \\
		\bottomrule
		
		Rental caching: fault model & \multicolumn{4}{|c|}{Same as rental paging} \\
		\bottomrule
		
		Paging with zapping & Deterministic & & $\frac{N(2k + 1) - (k+1)}{N + 2k}$ & $\min (N, 2k+1)$ \\
		\bottomrule
		
		Weighted Paging with zapping & \multicolumn{4}{|c|}{Same as paging with zapping} \\
		\bottomrule
		
		Weighted Paging with zapping & \multicolumn{4}{|c|}{Same as paging with zapping} \\
		\bottomrule
		
		\multirow{4}{*}{Rental paging with zapping} & \multirow{3}{*}{Deterministic} & $\lambda \ge \inv{k}$ & & 3 \\
		\cmidrule{3-3} \cmidrule{5-5}
		& & $\inv{k^2} \le \lambda < \inv{k}$ & & $1+\frac{2}{k\lambda}$ \\
		\cmidrule{3-3} \cmidrule{5-5}
		& & $\lambda < \inv{k^2}$&  & $2k+1$ \\
		\bottomrule
		
		Weighted rental paging with zapping & Deterministic & & & $2k+1$ \\
		\bottomrule
		
		Rental caching with zapping & \multicolumn{4}{|c|}{Same as weighted rental paging with zapping} \\
		\bottomrule
		
		Rental caching with zapping: fault model & \multicolumn{4}{|c|}{Same as rental paging with zapping} \\
		\bottomrule
		
	\end{tabular}
\end{table}

Table \ref{tab:result-summary} presents the summary of the results in this paper.

For rental paging and for fault model, the deterministic upper and lower bounds in this paper are tight within constant factors. 
For the randomized case, the lower and upper bounds are tight within constant factors when $\lambda$ is $O(\inv{k^2 H_{k}})$ and when $\lambda \ge \inv{k}$. For weighted rental paging and for rental caching, the upper and lower bounds are tight within constant factors when $\lambda < \inv{k}$ for the deterministic case, and when $\lambda$ is $O(\inv{k^2 H_{k}})$ or when $\lambda \ge \inv{k}$ for the randomized case.
The bounds for the variants with rental cost are within constant factors of the bounds for the variants without rental cost when $\lambda \le \inv{k^2}$ for the deterministic case and when $\lambda$ is $O(\inv{k^2H_{k}})$ for the randomized case. For higher values of $\lambda \ge \inv{k}$ we show constant lower bounds and matching upper bounds.

For paging with zapping, weighted paging with zapping, and caching with zapping, the deterministic lower and upper bounds in this paper are tight within constant factors.

\subsection{Other work on rental paging}

\citet{salinger2012rental}, in an independent work, study the rental paging problem. They give a deterministic polynomial time algorithm for the offline problem by reducing it to interval weighted interval scheduling. They show that any \emph{conservative} or \emph{marking} algorithm is $k$-competitive and that the bound is tight. An algorithm is \emph{conservative} if it incurs at most $k$ faults on any consecutive subsequence of requests that contains at most $k$ distinct pages. A \emph{marking} algorithm marks each page when it is requested, and when it is required to evict a page, it evicts an unmarked page. If there are no unmarked pages, it first unmarks all the pages and then removes one.

For any online algorithm $A$ for paging, define the algorithm $A_d$ for rental paging as follows. $A_d$ behaves like $A$ with the modification that any page in the cache that has not been requested for $d$ steps is evicted. They define a class of online algorithms $M_d$, where $M$ is any conservative or marking algorithm. They show an upper bound of $2$ on the competitive ratio of $M_{\inv{\lambda}}$ when $\lambda > \inv{k}$, which matches the upper bound in this paper. They show an upper bound of $\max{(k, \frac{(k+1)}{1 + \lambda(k - 1)})}$ on the competitive ratio of $M_{\inv{\lambda}}$ when $\lambda \le \inv{k}$. This upper bound is weaker than the upper bound we present in this paper when $\inv{k^2} < \lambda < \inv{k}$.

Their deterministic lower bound on the competitive ratio for rental paging matches the lower bound in this paper.

They also present experimental results for the performance of various LRU, LRU$_{\inv{\lambda}}$, FWF, FWF$_{\inv{\lambda}}$, FIFO, FIFO$_{\inv{\lambda}}$, and the optimal offline algorithm. The experimental results agree with the upper bounds shown in the paper.

They present results only for rental paging and not for weighted rental paging or rental caching. They do not study the rental paging problem in the randomized setting.

\section{Online covering approach}
\label{sec:cilp-approach}

In this section, we give a brief overview of the online covering approach from \citep{young12greedy}. We use this approach, with modifications in some cases, to give deterministic algorithms for the variants of paging and caching problems in this paper. The idea is to reduce the given problem to online covering and then use the online covering algorithm from \citep{young12greedy} as follows. In online covering the constraints are revealed one at a time in any order. Whenever the algorithm gets a constraint that is not yet satisfied, it raises each variable in the constraint, at the rate inversely proportional to the coefficient of the variable in the cost function, until the constraint is satisfied. This algorithm is $\Delta$-competitive, where $\Delta$ is the maximum number of variables in any constraint.

Now we illustrate this approach for the case of paging. To formulate paging as a Covering Integer Linear Program (CILP), we define the following notation and continue using it in the remainder of the paper.
\begin{itemize}
	\item $f_t$ : file requested at time $t$
	\item $t^\prime$ : time of next request to the file requested at time $t$
	\item $x_t$ : indicator variable for the event that the file requested at time $t$ was evicted before $t^\prime$
	\item $R(t)$ : set of times of the most recent request to each file until and including time $t$
	\item $Q(t)$ : $\{ Q \subseteq R(t) - \{t\} : |Q| = k \}$. That is, $Q(t)$ represents all possible ways that the cache can be full when $f_t$ is requested at time $t$.
	\item $T$ : time of last request
\end{itemize}

We formulate paging as follows (LP-Paging):
\begin{align*}
	\min & \quad \sum_{t = 1}^{T}{x_t} \\
	\text{s.t.} \quad \forall t, \forall Q \in Q(t): & \quad \sum_{s \in Q}{\lfloor x_s \rfloor} \ge 1
\end{align*}

Each constraint represents the following. At time $t$, when $f_t$ is requested, for any subset $Q$ of $Q(t)$, it must be true that at least of the files, corresponding to the times in $Q$, must be evicted to make space for $f_t$. Clearly, any feasible solution to the paging problem, is a feasible solution to Paging-LP. In particular, any optimal solution to the paging problem, is a feasible solution to Paging-LP. For any variable $x$, $x^*$ denotes the value of $x$ in the optimal solution.

Now we describe the CILP based algorithm for paging. Note that, at each time step the algorithm may get multiple constraints. The algorithm considers the constraints in arbitrary order. When it gets a constraint that is not yet satisfied, it raises each variable om the constraint at unit rate until the constraint is satisfied. Whenever a variable reaches 1, the algorithm evicts the corresponding file from the cache. If the algorithm gets a constraint that is already satisfied, the algorithm does not do anything. We say, the algorithm does \emph{work} on a constraint, if it wasn't already satisfied and the algorithm raises the variables in the constraint, as described above, to satisfy it.

Each constraint in LP-Paging has exactly $k$ variables. Now we show that this algorithm is $k$-competitive using the following potential function.

$$ \phi = \sum_{t}{ \max (x^*_t - x_t) } $$

Initially, $\phi = \OPT$ and $\ALG = 0$. When the algorithm gets a constraint that is not satisfied, it raises each variable in the constraint at rate 1. So, the cost of the algorithm increases at the rate $k$. Also, $\phi$ decreases at unit rate because there is at least one variable $x_s$ in the constraint such that $x_s < x^*_s$ (otherwise the constraint would already be satisfied). Thus, the algorithm maintains the invariant $\ALG/k + \phi \le \OPT$. Since, $\phi \ge 0$, $\ALG \le k \cdot \OPT$.

For the variants in this paper, we use the approach outlined above, but with modifications in some cases. When we use the algorithm without any modifications, we omit the proofs for the competitive ratio. For these cases, the competitive ratio is the maximum number of variables in any constraint on which the algorithm does some work. If we apply any modifications, we present complete proofs.

\section{Rental caching}

\subsection{Deterministic algorithms using CILP}
\label{sec:rental-cilp}

In this section, we present a deterministic algorithm, \rentalpagingcilp, for rental paging, and then extend the algorithm to rental caching. Our algorithm is based on the greedy online covering algorithm outlined in Section \ref{sec:cilp-approach}. We use the notation defined in Section \ref{sec:cilp-approach}. In addition, we define the following indicator variable to account for renting files.
\begin{itemize}
	\item $y_{t, s}$ : indicator variable for the event that the file requested at time $t$ pays the rental cost at time $s < t^\prime$
\end{itemize}

The following is the formulation for rental paging (LP-Rental-Paging):
\begin{align*}
	\min & \quad \sum_{t = 1}^{T}{(x_t + \lambda \sum_{t \le s < t^\prime}{ y_{t, s} })} \\
	\text{s.t.} \quad \forall t, \forall Q \in Q(t): & \quad \sum_{s \in Q}{\lfloor x_s \rfloor} \ge 1 & (I) \\
	\quad \forall t, t \le s < t^\prime: & \quad \lfloor y_{t, s}  \rfloor + \lfloor x_t \rfloor \ge 1 & (II)
\end{align*}

The first set of constraints $(I)$ enforce the cache size at time $t$ (same as the constraints in LP-Paging), and the second set of constraints $(II)$ say that either a file has been evicted or it is being rented at time $s$. We denote them by \emph{cache-size} constraints and \emph{rent-evict} constraints, respectively.

For each request, \rentalpagingcilp gets some cache-size constraints and some rent-evict constraints. It considers the rent-evict constraints before the cache-size constraints. Whenever it gets a constraint that is not satisfied, it raises each variables in the constraint at the rate inversely proportional to its cost in the objective, until the constraint is satisfied. So, the algorithm raises $x_s$ at unit rate and $y_{t, s}$ at rate $\inv{\lambda}$.

For some $\gamma > 0$, \rentalpagingcilp$_\gamma$ is the algorithm that behaves like \rentalpagingcilp\ with the following modification. \rentalpagingcilp$_\gamma$ raises $y_{t, s}$ at the modified rate of $\frac{\gamma}{\lambda}$. Note that, \rentalpagingcilp$_1$ and \rentalpagingcilp\ are the same algorithm.

\begin{theorem}
	\label{thm:paging-rent-cilp}
	For rental paging, (a) \rentalpagingcilp\ is $2$-competitive when $\lambda \ge \inv{k}$, (b) \rentalpagingcilp$_{k \lambda}$ is $(1 + \inv{k \lambda})$-competitive when $\inv{k^2} < \lambda < \inv{k^2}$, and (c) \rentalpagingcilp is $k$-competitive when $\lambda \le \inv{k^2}$.
\end{theorem}

\begin{proof}
	(a) $\inv{k} \le \lambda$: We claim that, at any given time, if all the rent-evict constraints are satisfied, the cache-size constraints are satisfied too. We prove this by showing that each file is evicted within $k$ steps from its latest request, by considering just the rent-evict constraints.
	At any given time, the algorithm considers the rent-evict constraint corresponding to each file in the cache. In the rent-evict constraint at time $t$, when $y_{t, s}$ goes from 0 to 1, $x_s$ increases by $\lambda$. So, if the file has been in the cache for $t$ time steps since its latest request, $x_s = t \lambda$, which is at least 1 for $t \ge \inv{\lambda}$. Since, $\inv{\lambda} \le k$, $x_s$ will be 1 in at most $k$ steps. Thus, the algorithm does work only on rent-evict constraints, each of which has exactly 2 variables. So, \rentalpagingcilp\ is 2-competitive.

	(b) $\inv{k^2} < \lambda < \inv{k}$: When the algorithm considers a rent-evict constraints, it raises $x_s$ at unit rate, but raises $y_{t, s}$ at rate $\frac{\gamma}{\lambda}$, where $\gamma = k \lambda$. The increment in $x_s$ is $\frac{\lambda}{\gamma}$ at each time step. So, for $\gamma \le k \lambda$, within $k$ steps $x_s \ge 1$ and hence the corresponding file is evicted. Thus, like in the previous case, the algorithm never does any work on the cache-size constraints. Now we show that this algorithm is $(1 + \inv{k \lambda})$-competitive. The proof is similar to the proof in Section \ref{sec:cilp-approach}. We use the following potential function for our proof:

	$$ \phi = \sum_{t=1}^{T}{\Big( \max{(x^*_t - x_t, 0)}} + \sum_{t \le s < t^\prime}{\lambda \max{(y^*_{t, s} - y_{t, s}, 0)} \Big)}$$

	Consider the rent-evict constraint at time $s$ for the file whose most recent request was at time $t$. When the algorithm raises the variables in the constraint, the cost of the algorithm increases at the rate $(1 + \gamma)$. Also, $\phi$ decreases at the rate $\min (1, \gamma)$. Thus, the algorithm maintains the invariant $\ALG/(1 + \gamma) + \phi/(\min (1, \gamma)) \le \OPT$. It is true initially, because $\ALG = 0$ and $\phi = \OPT$. Since, $\phi \ge 0$, this implies that $\ALG \le \frac{1 + \gamma}{\min (1, \gamma)} \OPT$. Also, $\gamma = k \lambda \le 1$. So, $\ALG \le (1 + \inv{k \lambda}) \OPT$.
	
	(c) $\lambda \le \inv{k^2}$: In this case, \rentalpagingcilp\ does work on both cache-size constraints and rent-evict constraints, and thus, the algorithm is $k$-competitive.
\end{proof}

Now we extend the results to rental caching and present the algorithm \rentalcachingcilp. For rental caching, the linear program is similar to the linear program for rental paging, with appropriate changes to take into account the cost and the size of each file.
We define $Q(t)$ to take into account the file sizes as follows. $Q(t) = \{ Q \subseteq R(t) - \{t\} : k - size(f) < size(Q) \le k \}$, where $size(Q) = \sum_{t \in Q}{size(f_t)}$. We modify the objective to take into account the cost of files. The following is the formulation for rental caching (LP-Rental-Caching):
\begin{align*}
	\min & \quad \sum_{t = 1}^{T}{(cost(f_t) \cdot x_t + \lambda \sum_{t \le s < t^\prime}{ y_{t, s} })} \\
	\text{s.t.} \quad \forall t, \forall Q \in Q(t): & \quad \sum_{s \in Q}{size(f_s) \cdot \lfloor x_s \rfloor} \ge size(f_t) \\
	\quad \forall t, t \le s < t^\prime: & \quad \lfloor y_{t, s}  \rfloor + \lfloor x_t \rfloor \ge 1
\end{align*}

When \rentalcachingcilp\ gets a rent-evict constraint that is not yet satisfied, it raises $x_t$ at rate $\inv{cost(f_t)}$ and $y_{t, s}$ at rate $\inv{\lambda}$. When it gets a cache-size constraint, it raises $x_s$ at rate $\inv{cost(f_s)}$.

\begin{theorem}
	\label{thm:retnal-caching-cilp}
	\rentalcachingcilp\ is $k$-competitive for rental caching.
\end{theorem}

\begin{proof}
	Since each file has size at least 1, each constraint has at most $k$ variables. So, for the general case of rental caching, the algorithm is $k$-competitive.
\end{proof}

\begin{corollary}
	\label{cor:rental-fault-model-cilp}
	For rental caching for the case of fault model, (a) \rentalcachingcilp\ is $2$-competitive when $\lambda \ge \inv{k}$, (b) \rentalcachingcilp$_{k lambda}$ is $(1 + \frac{1}{k \lambda})$-competitive when $\inv{k^2} < \lambda < \inv{k^2}$, and (c) \rentalcachingcilp\ is $k$-competitive when $\lambda \le \inv{k^2}$.
\end{corollary}

\begin{proof}
	For the fault model, $cost(f)$ is 1 for each file $f$. So, the cost function and the rent-evict constraints are the same as in case of rental paging with zapping. Thus, the three cases of Theorem \ref{thm:paging-rent-cilp} still hold.
\end{proof}

\subsection{Rental caching with infinite cache}
\label{sec:rent-inf-cache}

Consider the special case of the rental paging (or caching) problem where the cache has infinite size. This is equivalent to the rental caching problem without any cache size constraint. Even though there is no cache size constraint, this problem is still interesting because there is a rental cost for keeping files in the cache.

\begin{theorem}
	\label{thm:rent-inf-cache}
	If there is an $\alpha$-competitive algorithm $\ALGSR$ for ski-rental, then there is a $(\frac{\lambda + \alpha}{\lambda + 1})$-competitive algorithm for rental caching with infinite cache.
\end{theorem}

\begin{proof}
	Consider any file $f$. We define a phase as follows. A phase starts with a request to $f$ and ends at the time step just before the next request to $f$. When a $f$ is requested, it is either already in the cache or it is retrieved and added to the cache. Thus, once a phase starts, the file must be present in the cache and the earliest this file can be evicted from the cache is at the next time step.
	
	Such a phase, excluding the first step, reduces to ski-rental as follows. The cost of renting is $\lambda$ and cost of buying is the cost of eviction, which is $cost(f)$. The algorithm doesn't know when the phase ends and at each time step it has to decide if it keep renting the file or if it should pay for the eviction cost to buy it.

	The algorithm $\ALGINF$ for rental caching with infinite cache does the following. For a request to file $f_t$ at time $t$, $\ALGINF$ brings the file into the cache. Starting at the next time step, it simulates $\ALGSR$ on $f_t$ to decide for how long it keeps the file in the current phase. If $\ALGSR$ buys $f$ at any time step during the phase, $\ALGINF$ evicts it from the cache at that step. The total rental cost of $\ALGINF$ is same as the total rental cost of $\ALGSR$ and the total eviction cost of $\ALGINF$ is equal to the total cost of buying for $\ALGSR$.

	Let $\OPTSR$ be the optimal cost of the ski-rental problem. In a phase, $\ALGINF$ cost is at most $\lambda + \alpha \cdot \OPTSR$ and the optimal cost is $\OPTINF = \lambda + \OPTSR$. The the competitive ratio of this algorithm is at most $\frac{\lambda + \alpha \cdot \OPTSR}{\lambda + \OPTSR}$. Since $\alpha > 1$, the competitive ratio is at most $\alpha$.

\end{proof}

\begin{corollary}
	\label{cor:rent-infinite-deterministic}
	$\ALGINF$ is a $2$-competitive deterministic algorithm for rental caching with infinite cache.
\end{corollary}

\begin{proof}
	The $2$-competitive deterministic algorithm for ski-rental \citep{karlin1988competitive} and Theorem \ref{thm:rent-inf-cache} together imply that $\ALGINF$ is $2$-competitive.
\end{proof}

\begin{corollary}
	\label{cor:rent-infinite-randomized}
	There is a $(\frac{e}{e-1})$-competitive randomized algorithm for rental caching with infinite cache.
\end{corollary}

\begin{proof}
	The $(\frac{e}{e-1})$-competitive randomized algorithm for ski-rental \citep{karlin1988competitive} and Theorem \ref{thm:rent-inf-cache} together imply that $\ALGINF$ is $(\frac{e}{e-1})$-competitive.
\end{proof}

\begin{theorem}
	\label{thm:high-rent-cache}
	If there is an $\alpha$-competitive algorithm $\ALGSR$ for ski-rental, then there is an $\alpha$-competitive algorithm for rental paging when $\lambda \ge \inv{k}$.
\end{theorem}


\begin{corollary}
	\label{cor:high-rent-cache-deterministic}
	When $\lambda \ge \inv{k}$, there is a $2$-competitive deterministic algorithm for rental caching.
\end{corollary}

\begin{corollary}
	\label{cor:high-rent-cache-randomized}
	When $\lambda \ge \inv{k}$, there is a $(\frac{e}{e-1})$-competitive randomized algorithm for rental caching.
\end{corollary}

\subsection{\rentalcachingmeta}

\begin{theorem}
	\label{thm:rent-meta-alg}
	If there is an $\alpha$-competitive algorithm $\ALGSR$ for ski-rental, and a $\beta$-competitive algorithm for caching (no rental cost) $\ALGC$, then there is $(\alpha + \beta)$-competitive algorithm for rental caching.
\end{theorem}

We present the \rentalcachingmeta\ algorithm. Our algorithm uses $\ALGSR$ and $\ALGC$ to generate a solution for rental caching. On an input sequence $\sigma$ and cache size $k$, \rentalcachingmeta\ does the following. It simulates $\ALGC$ on the input sequence $\sigma$ and cache $\cache{1}$ of size $k$. In parallel, it simulates $\ALGINF$ on the request sequence $\sigma$ and cache $\cache{2}$ of infinite size. $\ALGINF$ in turn simulates $\ALGSR$ on each request. At any time, the cache of \rentalcachingmeta\ contains the intersection of the files present in caches $\cache{1}$ and $\cache{2}$.

\begin{claim}
	The total size of the items in the cache of \rentalcachingmeta\ never exceeds $k$.
\end{claim}
	
\begin{proof}
	Total size of all items in the cache of $\ALGC$ is at least the total size of all items in the cache of \rentalcachingmeta. This proves our claim, because $\ALGC$ maintains the invariant that the total size of items in the cache is at most $k$.
\end{proof}

\begin{claim}
	$E[\text{\rentalcachingmeta}] \le E[\ALGSR] + E[\ALGC]$
\end{claim}

\begin{proof}
	\rentalcachingmeta\ evicts a file, when at least one of $\ALGSR$ and $\ALGC$ evicted the file. For each eviction, we charge the cost of eviction for \rentalcachingmeta\ to the algorithm that evicted the file, breaking ties arbitrarily. We charge the rental cost of \rentalcachingmeta\ to the rental cost of $\ALGSR$. This proves our claim.
\end{proof}

Also, $E[\ALGSR] \le \alpha \cdot \OPTSR \le \alpha \cdot \OPT$, and $E[\ALGC] \le \beta \cdot \OPTC \le \beta \cdot \OPT$, where $\OPTSR$ denotes the optimal cost for rental caching with infinite cache, $\OPTC$ denotes the optimal cost for caching, and $OPT$ denotes the optimal cost for rental caching. So, $E[\text{\rentalcachingmeta}] \le (\alpha + \beta) \OPT$, and hence, \rentalcachingmeta\ is $(\alpha + \beta)$-competitive algorithm for rental caching.

If both $\ALGC$ and $\ALGSR$ are deterministic, \rentalcachingmeta\ is also deterministic, otherwise it is randomized. Theorem \ref{thm:rent-meta-alg} implies the following two corollaries.
	
\begin{corollary}
	The $2$-competitive deterministic online algorithm for ski-rental \citep{karlin1988competitive} and the $k$-competitive deterministic online algorithm for caching \citep{sleator1985amortized}, give a $(k + 2)$-competitive deterministic online algorithm for rental caching.
\end{corollary}

\begin{corollary}
	The $(\frac{e}{e-1})$-competitive randomized online algorithm for ski-rental \citep{karlin1988competitive} and the $H_k$-competitive randomized online algorithm for caching \citep{mcgeoch1991strongly,achlioptas2000competitive}, give a $(H_k + \frac{e}{e-1})$-competitive randomized online algorithm for rental caching.
\end{corollary}

\subsection{Lower bounds}

\begin{theorem}
	The competitive ratio of any deterministic algorithm for rental paging is at least (a) $2$ when $\lambda > \inv{k}$, and (b) $\frac{k + k \lambda}{1 + k^2 \lambda}$ when $\lambda \le \inv{k}$.
\end{theorem}

\begin{proof}
	(a) $\lambda > \inv{k}$: Corollary \ref{cor:high-rent-cache-deterministic} implies a deterministic lower bound of $2$.

	(b) $\inv{k} \ge \lambda$: The adversary requests files from the set $\{1, 2, 3, \cdots, k+1\}$. At each step, the adversary requests a file that is not present in the cache of the algorithm.
	The algorithm faults at each time step and pays at least $\lambda$ at each step. OPT pays the rental cost to keep $k$ items in the cache at each time step and faults once in $k$ steps. So, the ratio is at least $\frac{k + k \lambda}{1 + k^2 \lambda}$. For sufficiently small $\lambda$, the ratio tends to $k$.
\end{proof}

\begin{theorem}
	The competitive ratio of any randomized algorithm for rental paging is at least (a) $\frac{e}{e-1}$ when $\lambda > \inv{k}$ and (b) $\frac{H_k + k H_k \lambda}{1 + k^2H_k \lambda}$ when $\lambda \le \inv{k}$.
\end{theorem}

\begin{proof}
	(a) $\lambda > \inv{k}$: Corollary \ref{cor:high-rent-cache-randomized} implies a randomized lower bound of $\frac{e}{e-1}$.

	(b) $\lambda \le \inv{k}$: The adversary requests files from a set of $k+1$. At each step, the adversary requests a file with uniform probability over all files except the file requested at the previous step. We split the request sequence into \emph{phases} as follows. A phase is the longest request sequence with at most $k$ distinct requests, and starts immediately after the previous phase ends.

	We now show that the expected length of each phase is $k H_k$. When $i$ files have been requested, the probability of requesting a file that has not been requested in the phase is $\frac{k - i}{k}$. Thus the expected length of a phase is $ \Big( \sum_{i=1}^{k}{\frac{k}{k-i+1}} \Big) = H_k$.

	Next, we show that the algorithm keeps its cache full to minimize the expected cost in a phase. Assume that $i$ distinct files have been requested in the phase and the algorithm has $p \le k$ files in the cache. At the next time step, the algorithm faults with a probability $\frac{k-c+1}{k}$ and pays a rental cost $c\lambda$. So, the expected cost of the algorithm is $1 + \inv{k} - p(\inv{k} - \lambda)$. Since, $\inv{k} \ge \lambda$, the cost is minimized when $p = k$. So, it pays $\inv{k}$ eviction cost and $k\lambda$ rental cost at each step. OPT pays the same rental cost, but faults once in each phase. So, for each phase, the ratio is at least $\frac{H_k + k^2 H_{k} \lambda }{1 + k^2 H_{k} \lambda }$. For sufficiently small $\lambda$, the ratio tends to $H_k$.
\end{proof}

\section{Caching with zapping}

\subsection{Deterministic algorithms using CILP}
\label{sec:zapping-cilp}

In this section, we present a deterministic algorithm, \zappingpagingcilp, for paging with zapping, and then extend the algorithm to caching with zapping. We introduce another indicator variable for zapping of files.

\begin{itemize}
	\item $z_f$: indicator variable for the event that the file $f$ has been zapped
\end{itemize}

We formulate paging with zapping as follows (LP-Paging-Zapping):
\begin{align*}
	\min & \quad \sum_{t = 1}^{T}{x_t} + N \sum_{f \in F} z_f \\
	\text{s.t.} \quad \forall t, \forall Q \in Q(t): & \quad \Big( \sum_{s \in Q}{\lfloor x_s \rfloor + \lfloor z_{f_s} \rfloor} \Big) + \lfloor z_{f_t} \rfloor \ge 1 \\
\end{align*}

The constraints say that either $f_t$ is zapped, or at least one file in the cache is either zapped or evicted. Whenever \zappingpagingcilp\ gets a constraint that is not satisfied, raises $x_s$ at unit rate and $z_{f_s}$ at rate $\inv{N}$.
	
\begin{theorem}
	\zappingpagingcilp\ is $(2k+1)$-competitive.
\end{theorem}

\begin{proof}
	Each constraint has $2k+1$ variables. Thus, \zappingpagingcilp\ is $(2k+1)$-competitive.
\end{proof}

Now we extend this approach to caching with zapping and present the algorithm \zappingcachingcilp. We define $Q(t)$ same as we defined for rental caching in section \ref{sec:rental-cilp}. We have the following linear program (LP-Caching-Zapping) for caching with zapping:

\begin{align*}
	\min & \quad \sum_{t = 1}^{T}{cost(f_t) \cdot x_t} + N \sum_{f \in F} z_f \\
	\text{s.t.} \quad \forall t, \forall Q \in Q(t): & \quad \Big( \sum_{s \in Q}{\min (\lfloor x_s \rfloor + \lfloor z_{f_s} \rfloor, 1) \cdot size(f_s) } \Big) + \lfloor z_{f_t} \rfloor \cdot size(f_t) \ge size(f_t) \\
\end{align*}

For each not yet satisfied constraint that \zappingcachingcilp\ gets, it raises $x_s$ at rate $\inv{cost(f)}$, and $z_{f_s}$ at rate $\inv{N}$.

\begin{theorem}
	\zappingcachingcilp\ is $(2k+1)$-competitive.
\end{theorem}

\begin{proof}
	Since each file has size at least 1, each constraint has at most $2k+1$ variables. Thus, the \zappingcachingcilp\ is $(2k+1)$-competitive.
\end{proof}

\begin{theorem}
	The algorithm that zaps the file when it is requested for the first time is $N$-competitive.
\end{theorem}

\begin{proof}
	Let the total number of distinct files requested be $T$. The total cost of OPT is at least $T$ (1 to bring each file in the cache). The total cost of the algorithm is at most $NT$. So, the ratio is at most $N$.
\end{proof}

\subsection{Lower Bounds}

\begin{theorem}
	For paging with zapping, the competitive ratio of any deterministic algorithm is at least $\frac{2Nk + N - (k+1)}{N + 2k}$.
\end{theorem}

\begin{proof}
	The adversary maintains a set of $k+1$ distinct files at all times. Every time a file is zapped by the algorithm, it is replaced in the set by a new file that has never been requested by the adversary. At each time step, the adversary requests a file that is not present in the cache of the algorithm. We define a \emph{zap-phase} as follows. A zap-phase ends every time a file is zapped and the following request marks the beginning of the next zap-phase. The first zap-phase starts with the first request of the input sequence. We define a \emph{round} as follows. The first round starts with the first request of the input sequence. A round ends when the algorithm has zapped the all of the $k+1$ files that were requested in the first $k+1$ time steps in that round (some other files may have been zapped too). The total number of files zapped in a round is at least $k+1$. The adversary repeats the process for a large number of rounds.

	Now we show the lower bound on the competitive ratio in each round. Consider any round. Let $T \ge k+1$ be the total number of files zapped by the algorithm in the round. Let $H_j$ be the length of zap-phase $j, 1 \le j \le T$.

	Any deterministic algorithm faults at each time step and zaps a total of $T$ files. So, the cost of any deterministic algorithm is at least, $ALG = NT + \sum_{j = 1}^{T}{(H_j - 1)} = (N-1) T + \sum_{j=1}^{T}{H_j}$. Note that, $\sum_{j=1}^{T}{H_j} \ge T$.

	When $\sum_{j=1}^{T}{H_j} = T$, the algorithm zaps each file when it is requested for the first time. In this case, adversary requests a new file at each step. OPT pays at most $\min{(1, N)}$ at each step while the algorithm pays $N$ at each step. For $N > 1$, the ratio is at least $N$.

	Now we assume that $\sum_{j=1}^{T}{H_j} > T$. Consider the offline algorithm $\zappingoff$ which, on any request sequence, does one of the following: (a) Does not zap any file, or (b) Chooses one file from the set of the first $k+1$ files requested and zaps it at the first step.

	If $\zappingoff$ doesn't zap any files, in the first zap-phase it pays $k$ to bring the first $k$ files into the cache and then pays at most $\lceil \frac{H_1 - k}{k} \rceil$ in the remainder of the first phase. For any zap-phase $j > 1$, $\zappingoff$ pays at most $\lceil \frac{H_j}{k} \rceil$.

	\begin{align*}
		\zappingoff 	& = (k + \lceil \frac{H_1 - k}{k} \rceil) +
						\Big( \sum_{j = 2}^{T}{\lceil \frac{H_j}{k} \rceil \Big) } \\
				& \le (k + \frac{H_1}{k}) + \Big( \sum_{j = 2}^{T}{\frac{H_j}{k} + 1} \Big) \\
				& = k + T - 1 + \sum_{j = 1}^{T}{\frac{H_j}{k}}
	\end{align*}

	If $\zappingoff$ zaps 1 file, it incurs $k$ faults in the first phase and 1 fault in each phase after that. It also pays $N$ for zapping 1 file. The total cost in this case is $(k) + (T - 1) + N$.

	Since $cost(\zappingoff) \ge OPT$, the competitive ratio is at least $$ \frac{\sum_{j = 1}^{T}{H_j} + (N-1) T}{\min{( k + T - 1 + \sum_{j = 1}^{T}{\frac{H_j}{k}} , k + T + N - 1)}}$$.
	The ratio is minimized when $k + T - 1 + \sum_{j = 1}^{T}{\frac{H_j}{k}} = k + T + N - 1$. Simplifying gives, $\sum_{j = 1}^{T}{H_j} = Nk$. So, the ratio is at least $$\frac{Nk + (N - 1) T}{N + k + T - 1}$$
	For any given $N$ and $k$, and for $T \ge k+1$, the ratio is minimized when $T = k+1$. So, the competitive ratio is at least
	\begin{align*}
		& \frac{2Nk + N - (k+1)}{N + 2k}
	\end{align*}

	Thus, the competitive ratio is at least $\min{(N, \frac{2Nk + N - (k+1)}{N + 2k})}$. Note that, $\frac{2Nk + N - (k+1)}{N + 2k} = N\frac{2k + 1}{N + 2k} - \frac{k+1}{N + 2k} \le N$, for $N \ge 1$. Thus the competitive ratio of any deterministic algorithm for paging with zapping is at least $\frac{2Nk + N - (k+1)}{N + 2k}$.
\end{proof}

\section{Rental paging with zapping}

\subsection{Deterministic algorithm using CILP}

In this section we present algorithm \rentalzappingpagingcilp\ for rental paging with zapping.

We use the same notations defined in the previous sections. The cache size constraints are exactly the same as in case of paging with zapping. The rent-evict constraints are modified to have variables for eviction, renting, and zapping. We have the following formulation (LP-Paging-Rental-Zapping):
\begin{align*}
	\min & \quad \sum_{t = 1}^{T}{(x_t} + \lambda \sum_{t \le s < t^\prime}{y_{t, s}}) + N \sum_{f \in F} z_f \\
	\text{s.t.} \quad \forall t, \forall Q \in Q(t): & \quad \Big( \sum_{s \in Q}{\lfloor x_s \rfloor + \lfloor z_{f_s} \rfloor} \Big) + \lfloor z_{f_t} \rfloor \ge 1 & (III) \\
	\quad \forall t, t \le s < t^\prime: & \quad \lfloor y_{t, s}  \rfloor + \lfloor x_t \rfloor + \lfloor z_{f_s} \rfloor \ge 1 & (IV)
\end{align*}

We refer to $(III)$ and $(IV)$ by cache-size constraints and \emph{rent-evict-zap} constraints, respectively. \rentalzappingpagingcilp\ is similar to \rentalpagingcilp, and considers the rent-evict-zap constraints before the cache-size constraints. Whenever the algorithm gets a constraint that is not satisfied, it raises all the variables in the constraint as follows. It raises $x_s$ at unit rate, $y_{t, s}$ at rate $\inv{\lambda}$, and $y_f$ at rate $\inv{N}$.

We define \rentalzappingpagingcilp$_\gamma$ as \rentalzappingpagingcilp\ with the following modification. \rentalzappingpagingcilp$_\gamma$ raises $y_{t, s}$ at the modified rate of $\frac{\gamma}{\lambda}$, where $\gamma > 0$.

\begin{theorem}
	\label{thm:paging-rent-zap-cilp}
	For rental paging with zapping, (a) \rentalzappingpagingcilp\ is $3$-competitive when $\lambda \ge \inv{k}$, (b) \rentalzappingpagingcilp$_{k \lambda}$ is $(1 + \frac{2}{k \lambda})$-competitive when $\inv{k^2} < \lambda < \inv{k^2}$, and (c) \rentalzappingpagingcilp\ is $(2k+1)$-competitive when $\lambda \le \inv{k^2}$.
\end{theorem}

\begin{proof}
	(a) $\inv{k} \le \lambda$: We claim that, at any given time, if all the rent-evict-zap constraints are satisfied, the cache-size constraints are satisfied too. We prove this by showing that each file is evicted within $k$ steps from its latest request, by considering just the rent-evict-zap constraints. At any given time, the algorithm considers the rent-evict-zap constraint corresponding to each file in the cache.

	Consider the rent-evict-zap constraint at time $s$ for the file whose latest request was at time $t$.
	For this constraint, when $y_{t, s}$ goes from 0 to 1, $x_t$ increases by $\lambda$. So, the file has been in the cache for $s$ time steps since its latest request and $x_t = s \lambda$. This value is at least 1 if the file has been in the cache for $s \ge \inv{\lambda}$ steps. Since, $k \ge \inv{\lambda}$, $x_s$ is at least 1 when $s = k$, and consequently all the cache-size constraints where $x_s$ participates will be satisfied. Note that, if $z_{f_t}$ is 1 before $x_t$ is 1, the cache-size constraints are still satisfied.
	Thus, the algorithm does work only on rent-evict-zap constraints, each of which has exactly 3 variables. Thus, \rentalzappingpagingcilp\ is 3-competitive.

	(b) $\inv{k^2} < \lambda < \inv{k}$: When the algorithm considers a rent-evict-zap constraints, it raises $x_s$ at unit rate, but raises $y_{t, s}$ at rate $\frac{\gamma}{\lambda}$, where $\gamma = k \lambda$. The increment in $x_s$ is $\frac{\lambda}{\gamma}$ at each time step. So, for $\gamma \le k \lambda$, within $k$ steps $x_s \ge 1$ and hence the corresponding file is evicted. Thus, like in the previous case, the algorithm never does any work on the cache-size constraints. Now we show that this algorithm is $(1 + \inv{k \lambda})$-competitive. We use the following potential function for our proof:

	$$ \phi = \sum_{t=1}^{T}{\Big( \max{(x^*_t - x_t, 0)}} + \sum_{t \le s < t^\prime}{\lambda \max{(y^*_{t, s} - y_{t, s}, 0)} \Big) + \sum_{f \in F}{N \max{(z^*_f - z_f, 0)}} }$$
	
	Consider the rent-evict-zap constraint at time $s$ for the file whose most recent request was at time $t$.

	When the algorithm raises the variables in the constraint, the cost of the algorithm increases at the rate $(2 + \gamma)$. Also, $\phi$ decreases at the rate $\min (1, \gamma)$. Thus, the algorithm maintains the invariant $\ALG/(2 + \gamma) + \phi/(\min (1, \gamma)) \le \OPT$. It is true initially, because $\ALG = 0$ and $\phi = \OPT$. Since, $\phi \ge 0$, this implies that $\ALG \le \frac{2 + \gamma}{\min (1, \gamma)} \OPT$. Also, $\gamma = k \lambda \le 1$. So, $\ALG \le (1 + \frac{2}{k \lambda}) \OPT$.

	(c) $\lambda \le \inv{k^2}$: In this case, \rentalzappingpagingcilp\ does work on both cache-size constraints and rent-evict-zap constraints, and thus, the algorithm is $(2k+1)$-competitive.
\end{proof}

To extend the algorithm to rental caching with zapping, we combine the ideas from Sections \ref{sec:rental-cilp} and \ref{sec:zapping-cilp}. We modify $Q(t)$ to account for file sizes. $Q(t) = \{ Q \subseteq R(t) - \{t\} : k - size(f) < size(Q) \le k \}$. The following is the formulation of rental caching with zapping (LP-Caching-Rental-Zapping):
\begin{align*}
	\min & \quad \sum_{t = 1}^{T}{\Big( cost(f_t) \cdot x_t} + \lambda \sum_{t \le s < t^\prime}{y_{t, s}} \Big) + N \sum_{f \in F} z_f \\
	\text{s.t.} \quad \forall t, \forall Q \in Q(t): & \quad \Big( \sum_{s \in Q}{\min (\lfloor x_s \rfloor + \lfloor z_{f_s} \rfloor, 1) \cdot size(f_s) } \Big) + \lfloor z_{f_t} \rfloor \cdot size(f_t) \ge size(f_t) \\
	\quad \forall t, t \le s < t^\prime: & \quad \lfloor y_{t, s}  \rfloor + \lfloor x_t \rfloor + \lfloor z_{f_s} \rfloor \ge 1 & (IV)
\end{align*}

When \rentalzappingcachingcilp\ gets a rent-evict-zap constraint that is not yet satisfied, it raises $x_t$ at rate $\inv{cost(f_t)}$, $y_{t, s}$ at rate $\inv{\lambda}$, and $z_{f_t}$ at rate $\inv{N}$. When it gets a cache-size constraint, it raises $x_s$ at rate $\inv{cost(f_s)}$ and $z_{f_s}$ at rate $\inv{N}$.

\begin{theorem}
	\label{thm:rental-zap-caching-cilp}
	\rentalzappingcachingcilp\ is $(2k+1)$-competitive for rental caching with zapping.
\end{theorem}

\begin{proof}
	Since each file has size at least 1, each cache-size constraint has at most $(2k+1)$ variables and each rent-evict-zap constraint has exactly 3 variables. So, for the general case of caching with zapping, the algorithm is $(2k+1)$-competitive.
\end{proof}

\begin{corollary}
	\label{cor:rental-zapping-fault-model-cilp}
	For rental caching with zapping for the case of fault model, (a) \rentalzappingcachingcilp\ is $3$-competitive when $\lambda \ge \inv{k}$, (b) \rentalzappingcachingcilp$_{k \lambda}$ is $(1 + \frac{2}{k \lambda})$-competitive when $\inv{k^2} < \lambda < \inv{k^2}$, and (c) \rentalzappingcachingcilp\ is $(2k+1)$-competitive when $\lambda \le \inv{k^2}$.
\end{corollary}

\begin{proof}
	For the fault model, $cost(f)$ is 1 for each file $f$. So, the cost function and the rent-evict-zap constraints are the same as in case of rental paging with zapping. Thus, the three cases of Theorem \ref{thm:paging-rent-zap-cilp} still hold.
\end{proof}

\subsection{\rentalzappingcachingmeta}

Analogous to Theorem \ref{thm:rent-meta-alg}, we have the following theorem for rental caching with zapping.

\begin{theorem}
	\label{thm:rent-zap-meta-alg}
	If there is an $\alpha$-competitive algorithm $\ALGSR$ for ski-rental, and a $\beta$-competitive algorithm for caching with zapping (no rental cost) $\ALGZ$, then there is $(\alpha + \beta)$-competitive algorithm for caching with zapping and rental cost.
\end{theorem}

We present the \rentalzappingcachingmeta\ algorithm. Our algorithm uses $\ALGSR$ and $\ALGZ$ to generate a solution for rental caching with zapping. On an input sequence $\sigma$ and cache size $k$, \rentalzappingcachingmeta\ does the following. It simulates $\ALGZ$ on the input sequence $\sigma$ and cache $\cache{1}$ of size $k$. In parallel, it simulates $\ALGINF$ on the request sequence $\sigma$ and cache $\cache{2}$ of infinite size. $\ALGINF$ in turn simulates $\ALGSR$ on each request. At any time, the cache of \rentalzappingcachingmeta\ contains the intersection of the files present in caches $\cache{1}$ and $\cache{2}$. If $\ALGZ$ nukes a file, \rentalzappingcachingmeta\ nukes it.

\begin{claim}
	The total size of the items in the cache of \rentalzappingcachingmeta\ never exceeds $k$.
\end{claim}

\begin{proof}
	Total size of all items in the cache of $\ALGZ$ is at least the total size of all items in the cache of \rentalzappingcachingmeta. This proves our claim, because $\ALGZ$ maintains the invariant that the total size of items in the cache is at most $k$.
\end{proof}

\begin{claim}
	$E[\text{\rentalzappingcachingmeta}] \le E[\ALGSR] + E[\ALGZ]$
\end{claim}

\begin{proof}
	\rentalzappingcachingmeta\ evicts a file, when at least one of $\ALGSR$ and $\ALGZ$ evicted the file. For each eviction, charge the cost of eviction for \rentalzappingcachingmeta\ to the algorithm that evicted the file, breaking ties arbitrarily. Charge the cost of zapping to $\ALGZ$ and charge the rental cost to the rental cost of $\ALGSR$. This proves our claim.
\end{proof}

Also, $E[\ALGSR] \le \alpha \cdot \OPTSR \le \alpha \cdot \OPT$, and $E[\ALGZ] \le \beta \cdot \OPTZ \le \beta \cdot \OPT$, where $\OPTSR$ denotes the optimal cost for rental caching with infinite cache, $\OPTZ$ denotes the optimal cost for caching with zapping, and $OPT$ denotes the optimal cost for rental caching with zapping. So, $E[\text{\rentalzappingcachingmeta}] \le (\alpha + \beta) \OPT$, and hence, \rentalzappingcachingmeta\ is $(\alpha + \beta)$-competitive algorithm for rental caching. If both $\ALGZ$ and $\ALGSR$ are deterministic, \rentalcachingmeta\ is also deterministic, otherwise it is randomized.

\section{Conclusions and further directions}

We present lower and upper bounds, in deterministic and randomized settings, for rental paging and rental caching. For most cases, the lower and upper bounds are tight up to constant factors. When $\inv{k^2 H_k} < \lambda < \inv{k}$, there is a gap between the randomized lower and upper bounds shown in this paper.
The lower bounds in this paper assume that the cache of OPT is always full, and consequently, in each phase is OPT's rental cost is no longer O(OPT's eviction cost). It may be possible to show better lower bounds using a modified analysis or possibly by using another adversary strategy.

The deterministic lower and upper bounds for paging with zapping are tight up to constant factors. The next step would be to study randomized lower bounds and randomized algorithms for caching with zapping.

For rental caching with zapping, we present the upper bounds in both deterministic and randomized settings. It would be interesting to study the lower bounds and .

The models in this paper assume uniform rental cost and uniform zapping cost in this study. Note that, in our model for rental caching, the total rental cost depends only on the size of a file. A natural extension would be to consider models with (arbitrary) non-uniform rental and zapping costs.

The CILP based approach by \citet{young12greedy} is a general and elegant approach for deriving deterministic algorithms for online paging and caching problems. We use this approach for all the new variants studied in this paper. The algorithms thus derived may not be optimal, as we show for the case of rental paging (or caching) and also for rental paging (or caching) with zapping. For the problems in this paper, we were able to apply simple modifications to achieve upper bounds within constant factors on the lower bounds.

The primal-dual approach in \citep{bansal2007primal, bansal2008randomized, buchbinder2009online, adamaszek2012log, epstein2011variants} is a powerful framework for deriving randomized algorithms for online caching problems. It would be interesting to investigate if their approach can be used to give randomized algorithms for the variants studied in this paper.

\section{Acknowledgments}

We would like to thank Marek Chrobak and Li Yan for useful discussions.

\bibliographystyle{abbrvnat}
\bibliography{references}

\end{document}